\documentclass[runningheads]{llncs}
\usepackage{amsmath,amssymb,amsfonts, mathtools}
\usepackage{algorithm, algpseudocode}
\usepackage{graphicx}
\usepackage{textcomp}
\usepackage{xcolor}
\usepackage{booktabs} % For formal tables
\usepackage{hyperref}
\usepackage{amsmath}
\usepackage{enumitem}
\usepackage{relsize}
\usepackage{subcaption}
\usepackage{multirow}
\usepackage[depth=3]{bookmark}

\DeclarePairedDelimiter{\ceil}{\lceil}{\rceil}

\begin{document}
\title{mmLSH: A Practical and Efficient Technique for Processing Approximate Nearest Neighbor Queries on Multimedia Data}
\titlerunning{mmLSH: An Efficient Technique for Processing ANNS on Multimedia Data}
\author{Omid Jafari\orcidID{0000-0003-3422-2755} \and
Parth Nagarkar\orcidID{0000-0001-6284-9251} \and
Jonathan Monta\~no\orcidID{0000-0002-5266-1615}
}
\authorrunning{O. Jafari, et al.}
\institute{New Mexico State University, Las Cruces, US \and
\email{\{ojafari, nagarkar, jmon\}@nmsu.edu}}
\maketitle              % typeset the header of the contribution
\begin{abstract}
Many large multimedia applications require efficient processing of nearest neighbor queries. 
Often, multimedia data are represented as a collection of important high-dimensional feature vectors. Existing Locality Sensitive Hashing (LSH) techniques require users to find top-k similar feature vectors for each of the feature vectors that represent the query object. This leads to wasted and redundant work due to two main reasons: 1) not all feature vectors may contribute equally in finding the top-k similar multimedia objects, and 2) feature vectors are treated independently during query processing. Additionally, there is no theoretical guarantee on the returned multimedia results. In this work, we propose a practical and efficient indexing approach for finding top-k approximate nearest neighbors for \underline{m}ulti\underline{m}edia data using \underline{LSH} called \textit{mmLSH}, which can provide theoretical guarantees on the returned multimedia results. Additionally, we present a buffer-conscious strategy to speed up the query processing. Experimental evaluation shows significant gains in performance time and accuracy for different real multimedia datasets when compared against state-of-the-art LSH techniques.

\end{abstract}

\section{Introduction}
Finding nearest neighbors in high-dimensional spaces is an important problem in several multimedia applications. In multimedia applications, content-based data objects, such as images, audio, videos, etc., 
are represented using high-dimensional feature vectors, which are extracted using feature extraction algorithms. Locality Sensitive Hashing (LSH) \cite{Gionis:1999:SSH:645925.671516} is one of the most popular solutions for the approximate nearest neighbor (ANN) problem in high-dimensional spaces. Since it was first introduced in \cite{Gionis:1999:SSH:645925.671516}, many variants of LSH have been proposed \cite{Gan:2012:LHS:2213836.2213898,Huang:2015:QLH:2850469.2850470,Liu:2019,Tobias:2019}
that mainly focused on improving the search accuracy and/or the search performance of the given queries. LSH is known for two main advantages: its sub-linear query performance (in terms of the data size) and theoretical guarantees on the query accuracy. While the original LSH index structure suffered from large index sizes (in order to obtain a high query accuracy), state-of-the-art LSH techniques \cite{Gan:2012:LHS:2213836.2213898,Huang:2015:QLH:2850469.2850470} have alleviated this issue by using advanced methods such as \textit{Collision Counting} and \textit{Virtual Rehashing}. Thus, owing to their small index sizes, fast index maintenance, fast query performance, and theoretical guarantees on the query accuracy, we propose to build \textit{mmLSH} upon existing state-of-the-art LSH techniques. 

\noindent\textbf{Motivation of our work: Drawbacks of LSH on Multimedia Data.}\\
Popular feature extraction algorithms, such as SIFT, SURF (for images), Marsyas (for audio), etc., extract multiple features that collectively represent the object of interest for improved accuracy during retrieval. Hence, if a user wants to find similar objects to a given query object, nearest-neighbor queries have to be performed for every individual feature vector representing the query object (and then these intermediate results are aggregated to find the final object results (Section \ref{sec:mmLSH})). Existing techniques treat these individual feature vectors as independent of each other, and hence cannot leverage common elements between these feature vector queries for improved query performance. Most importantly, existing techniques can only give theoretical guarantees on the accuracy of the individual feature vector queries, but not on the final object results, unlike our proposed index structure, \textit{mmLSH}. 

\noindent\textbf{Contributions of this Paper:} In this paper, we propose a practical and efficient indexing approach for finding top-k approximate nearest neighbors for \underline{m}ulti\underline{m}edia data using \underline{LSH}, called \textit{mmLSH}. To the best of our knowledge, we are the first work to provide a rigorous theoretical analysis for answering approximate nearest neighbor queries on high-dimensional multimedia data using LSH. Our main contributions are:
\begin{itemize}[leftmargin=*]
	\item \textit{mmLSH} can efficiently solve approximate nearest neighbor queries for multimedia data while providing  rigorous theoretical analysis and guarantees on the accuracy of the query result.
	\item Additionally, we present an advanced buffer-conscious strategy to speedup the processing of a multimedia query. 
	\item Lastly, we experimentally evaluate \textit{mmLSH}, on diverse real multimedia datasets and show that \textit{mmLSH} can outperform the state-of-the-art solutions in terms of performance efficiency and query accuracy. 
\end{itemize}

\section{Related Work}
\label{sec:relWork}

LSH was originally proposed in \cite{Gionis:1999:SSH:645925.671516} for the Hamming distance and then later extended to the popular Euclidean distance \cite{Datar:2004:LHS:997817.997857}. C2LSH \cite{Gan:2012:LHS:2213836.2213898} introduced two main concepts of \textit{Collision Counting} and \textit{Virtual Rehashing} (explained in Section \ref{sec:prelim}) that solved the two main drawbacks of E2LSH \cite{Datar:2004:LHS:997817.997857}. QALSH \cite{Huang:2015:QLH:2850469.2850470} used these two concepts to build query-aware hash functions such that the hash value of the query object is considered as the anchor bucket during query processing. \cite{Sundaram:2013:SSS:2556549.2556574} proposes an efficient distributed LSH implementation which includes a cache-conscious hash table generation (to avoid cache misses to improve the index construction time). Our proposed cache-conscious optimization is to improve the efficiency of the query processing (and hence very different). \\
\textbf{Query Workloads in High-Dimensional Spaces: }Until now, only two works \cite{Nagarkar:CIKM,Jafari:2019:QCI:3323873.3325048} have been proposed that focus on efficient execution of query workloads in high-dimensional spaces. Neither of these two works provide a rigorous theoretical guarantees on the accuracy of the final result. In \cite{Nagarkar:CIKM}, the authors propose to efficiently execute set queries using a two-level index structure. The problem formulation, which is quite restrictive compared to our work, states that a point will be considered in the result set only if it satisfies a certain user-defined percentage of the queries in the query workload. In \cite{Jafari:2019:QCI:3323873.3325048}, the authors build a model based on the cardinality and dimensionality of the high-dimensional data to efficiently utilize the cache. The main drawback of these two approaches is that they require prior information that is found by analyzing past datasets. Hence the accuracy and efficiency of the index structures is determined by the accuracy of the models. Our proposed work is very different from these previous works: \textit{mmLSH} does not require any training models and additionally, we provide a theoretical guarantee on the accuracy of the returned results. 

\section{Key Concepts and Problem Specification}
\label{sec:prelim}

A \textit{hash function family} $H$ is ($R$, $cR$, $p_1$, $p_2$)-sensitive if it satisfies the following conditions for any two points $x$ and $y$
in a $d$-dimensional dataset $D \subset \mathbb{R}^d$: 
if $|x - y| \leq R$, then $Pr[h(x) = h(y)] \geq p_1$, and if $|x - y| > cR$, then $Pr[h(x) = h(y)] \leq p_2$.
Here, $p_1$ and $p_2$ are probabilities and $c$ is an approximation ratio. LSH requires that $c > 1$ and $p_1 > p_2$. 
In the original LSH scheme for Euclidean distance, each hash function is defined as $h_{\vec{a},b} (x) = \left\lfloor{\frac{\vec{a}.x + b}{w}}\right\rfloor,$ where $\vec{a}$ is a $d$-dimensional random vector and $b$ is a real number chosen uniformly from $[0, w)$, such that $w$ is the width of the hash bucket \cite{Datar:2004:LHS:997817.997857}.  
\textbf{C2LSH} \cite{Gan:2012:LHS:2213836.2213898} showed that two close points $x$ and $y$ collide in at least $l$ hash layers with a probability $1-\delta$, when the total number,
$m$, of hash layers are equal to:
$m = \ceil[\big]{\frac{\ln(\frac{1}{\delta})}{2(p_1-p_2)^2}(1+z)^2}$.

Given a multidimensional database $\mathcal{D}$, $\mathcal{D}$ consists of $n$ $d$-dimensional points that belongs to  $\mathbb{R}^d$. Each $d$-dimensional point $x_i$ is associated with an object $X_j$ s.t. multiple points are \textit{associated} with a single object. 
There are $S$ objects in the database ($1 \leq S \leq n$), and for each object $X_j$, $set(X_j)$ denotes the set of points that are associated with $X_j$. Thus, $n = \sum_{j=1}^{S} |X_j|$. 

Our goal is to provide a $k$-NN version of the $c$-approximate nearest neighbor problem for multidimensional objects. For this, we propose a notion of distance between multidimensional objects called {\it $\Gamma$-distance} (defined in Section \ref{sec:keydef}) and  that depends on  a percentage parameter that we denote by $\Gamma$.

Let us denote the $\Gamma$-distance between two objects $X_1$ and $X_2$ by $\Gamma dist(X_1, X_2)$. For a given query object $Q$, an object $X_j$ is a {\it $\Gamma$-$c$-approximate nearest neighbor of $Q$} if the $\Gamma$-distance between $Q$ and $X_j$ is at most $c$ times the $\Gamma$-distance between $Q$ and its true (or exact) nearest neighbor, $X_j^*$, i.e. $\Gamma dist(Q, X_j) \leq c\times \Gamma dist(Q, X_j^*)$, where $c>1$ is an \textit{approximation ratio}. Similarly, the $\Gamma k$-NN version of this problem states that we want to find $k$ objects that are respectively the $\Gamma$-$c$-approximate nearest neighbors of the exact $k$-NN objects of $Q$.

\section{\lowercase{mm}LSH}
\label{sec:mmLSH}
The Borda Count method \cite{borda:doi:10.1177/0192512102023004002} (along with other aggregation techniques \cite{borda:PEREZ2011951}) are popular existing techniques to aggregate results of multiple point queries to find similar objects in multimedia retrieval \cite{Arora:2018:HPS:3204028.3228393}. In order to find top-$k$ nearest neighbor \textit{objects} of multimedia object query $Q$, the existing methods find the top-$k'$ nearest neighbor \textit{points} for each query \textit{point} $q_i$, where $1\leq i \leq |set(Q)|$, $k$ is the number of desired results by the user, and $k'$ is an arbitrarily chosen number such that $k' >> k$ \cite{Arora:2018:HPS:3204028.3228393}. Once the top-$k'$ nearest neighbors of each query point $q_i$ is found, an overall score is assigned to each multimedia object $X_j$ based on the depth of the points (associated with $X_j$) in the top-$k'$ results for each of the point queries $q_i$ of $Q$. \textbf{Drawbacks of this approach: } 1) there is no theoretical guarantee for the accuracy of the returned top-$k$ result objects, and 2) all query points $q_i$ of the query object $Q$ are executed independently of each other. Hence, if a query point takes too long to execute as compared to others, then the overall processing time is negatively affected. Our proposed method, \textit{mmLSH}, solves both these drawbacks as explained in the next sections.

\subsection{Key Definitions of mmLSH}
\label{sec:keydef}

\subsubsection{Justification for using $R$-Object Similarity and $\Gamma$-distance: } In order to define two \textit{Nearby Objects}, we first define a similarity/distance measure between two objects in the context of ANN search. Note that, there have been several works that have defined voting-based similarity/distance measures between two multimedia objects, especially images \cite{Zhou:2011:LSI:2072298.2072012,Jegou:2010:IBL:1718320.1718326,Wu:2009}. 
Also, region-based algorithms have been explored in the past, whose main strategy is to divide the query object in regions and compare these with regions of the dataset objects via some region distance, and then aggregate the resulting distances \cite{BPM}. In this work we define the $\Gamma$-distance as a way to measure distances between objects as a whole.  Our definition follows the naive strategy of comparing all pairs of features of the objects but it  uses a percentage parameter  $\Gamma$ to ensure two identical objects have a near zero  distance. Another key advantage of the proposed distance is that it allows us to provide theoretical guarantees for our results. 

\begin{definition}[$R$-Object Similarity]
	Given a radius R, the R-Object Similarity between two objects $Q$ and $X_j$, that consists of $set(Q)$ and $set(X_j)$ $d$-dimensional feature vectors respectively, is defined as: 
	\begin{equation}
	sim(Q, X_j, R) = \frac{|\{q \in set(Q), x_i \in set(X_j) : \; ||q, x_i|| \leq R\}|}{|set(Q)|.|set(X_j)|}
	\end{equation}  
\end{definition}

Note that, $0 \leq sim(Q, X_j, R) \leq 1$. $sim(Q, X_j, R)$ will be equal to 1 if every point of $Q$ is a distance at most $R$ to every point of  $X_j$ (e.g. if you are comparing an entirely green image with another green image - and assuming the feature vectors were based on the color of the pixel. But if you are comparing  two identical images, then $sim(Q, X_j, R) < 1$ if $R$ is less than the largest among $||q,x_i||$). Since the number of points associated with two objects can be different, we normalize the similarity w.r.t the points associated with $Q$ and $X_j$. 

\begin{definition}[$\Gamma$-distance]
	Given a  two objects $Q$ and $X_j$, the $\Gamma$-distance between $Q$ and $X_j$ is defined as: 
	\begin{equation}
	\label{eqn:gamma_dist}
	\Gamma dist(Q, X_j) = \inf\{R\mid sim(Q, X_j, R)
	\geq \Gamma\} 
	\end{equation}  
\end{definition}

In order to find points that are within $R$ distance, we use the \textit{Collision Counting} method that is introduced in C2LSH \cite{Gan:2012:LHS:2213836.2213898}: given a query point $q$ and $m$ projections, a point $x$ is considered a \textit{candidate} if $x$ collides with $q$ (also called the \textit{collision count}) in at least $l$ projections (see Section \ref{sec:prelim}). \cite{Gan:2012:LHS:2213836.2213898} proves that if $||q, x|| \leq R$, then the collision count of $x$ with respect to $q$ (denoted by $cc(q, x)$) will be at least $l$ with a success probability of $1-\delta$.

We define a \textit{Collision Index} (denoted by  $ci(Q, X_j)$ that determines how close two objects are based on
the number of points between the two objects that are considered \textit{close} (i.e. the collision counts between the points of the two objects is greater than the collision threshold $l$).

\begin{definition}[Collision Index of Two Objects]
	Given two objects $Q$ and $X_j$, the collision index of $X_j$ with respect to $Q$ is defined as: 
	\begin{equation}
	\label{eqn:ci}
	ci(Q, X_j) = \frac{|\{q \in set(Q), x_i \in set(X_j) : \; cc(q, x_i) \geq l\}|}{|set(Q)|.|set(X_j)|}
	\end{equation}
\end{definition}
The \textit{Collision Index} between two objects depends on how many nearby points are considered as candidates between the two objects.  Thus, in turn, the accuracy of the collision index depends on the accuracy of the collision counting process (which is shown to be very high \cite{Gan:2012:LHS:2213836.2213898,Huang:2015:QLH:2850469.2850470}). Hence we define an object $X_j$ to be a {\it $\Gamma$-candidate} if the collision index between them is  greater than or equal to $(1-\varepsilon)\Gamma$, where $\varepsilon>\delta$ is an approximation factor which we set to $2\delta$.

\begin{definition}[$\Gamma$-candidate Objects]
	Given  an object query $Q$ and an object  $X_j$, we say that $X_j$ is a $\Gamma$-candidate with respect to $Q$ if $ci(Q, X_j) \geq (1-\varepsilon)\Gamma$. 
\end{definition}

Additionally, we define an object to be a {\it $\Gamma$-false positive } if it is a $\Gamma$-candidate but its $\Gamma$-distance to the object query is too high. 
\begin{definition}[$\Gamma$-False Positives]
	Given an object query $Q$ and an object $X_j$, we say $X_j$ is a \textit{$\Gamma$-false positive} with respect to $Q$ if we have $ci(Q, X_j) \geq \Gamma+\frac{\beta}{2}$ but $\Gamma dist(Q, X_j)>cR$.
\end{definition}

\begin{algorithm}[t]                      % enter the algorithm environment
	\caption{k-Nearest Neighbor Object}
	\label{alg:knn} 
	\begin{algorithmic}[1]                    % enter the algorithmic environment
		\While{TRUE}
		\If{$|\{X_j | X_j \in \mathcal{CL} \wedge \Gamma dist(Q, X_j) \leq cR\} | \geq k$}
		\State return the top-$k$ objects from $\mathcal{CL}$;
		\EndIf
		
		\For{$g=1; \;\;g\leq m; \;\;g++$}
		\For{$i=1; \;\; i \leq |set(Q)|; \;\; i++$}
		\State $CountCollisions(q_i)$;
		\State $\forall_{X_j\in S} $ Update $ci(Q, X_j)$;
		\EndFor
		
		\If{$|\mathcal{CL}| \geq k+\beta S$}
		\State return the top-$k$ objects from $\mathcal{CL}$;
		\EndIf 	
		
		\EndFor
		\State $R = c^{numIter}$;
		\State $numIter++$;
		\EndWhile
%		\EndProcedure
	\end{algorithmic}
%\vspace*{-0.05in}	
\end{algorithm}

\subsection{Design of mmLSH}
\label{sec:design}

During query processing, instead of executing the query points of $Q$ independently, we execute them one at a time in each projection (Lines 5-6 in Algorithm \ref{alg:knn}). The function $CountCollisions(q_i)$ (Line 7), an existing function from C2LSH, is responsible for counting collisions of query points and points in the database. The \textit{Buffer-conscious Optimizer} module (Section \ref{sec:bufferStrategy}) is responsible for finding an effective strategy to utilize the buffer to speed up the query processing. This module decides which query and the hash bucket should be processed next.
The \textit{$\Gamma$-Analyzer} module is in charge of calculating the \textit{collision indexes} (Section \ref{sec:keydef}) for objects in the database and for checking/terminating the process if the terminating conditions are met. 

\noindent\textbf{Terminating Conditions for mmLSH:}
The existing solution (Section \ref{sec:mmLSH}) finds top-$k'$ candidates for each query point in $Q$ and then terminates. Instead, \textit{mmLSH} stops when top-$k$ objects are found. These conditions  guarantee that $\Gamma$-$c^2$-approximate NN are found with constant probability (Section \ref{sec:the}):

\begin{enumerate}[label=$\mathcal{T}\arabic*)$]
	
	\item At certain point at level-$R$,  at least $k + \beta S$ $\Gamma$-candidates have been found, where  $\beta S$ is the allowed number of false positives. (Line 14, Algorithm \ref{alg:knn})
	
	\item At the end  of level-$R$, there exists at least $k$  $\Gamma$-candidates  whose $\Gamma$-distance to $Q$ is at most $R$. (Line 6, Algorithm \ref{alg:knn})

\end{enumerate}

\subsection{Buffer-conscious Optimization for Faster Query Processing}
\label{sec:bufferStrategy}

Another goal of \textit{mmLSH} is to improve the processing speed of finding nearest neighbors of a given multimedia query object by efficiently utilizing a given buffer space. In order to explain our strategy, we first analyze the two expensive operations and the two naive strategies for solving the problem. The two main dominant costs in LSH-based techniques are the \textit{Algorithm time} (which is the time required to find the candidate points that collide with the given query point) and the \textit{Index IO time} (which is the time needed to bring the necessary index files from the secondary storage to the buffer).

\noindent 
Due to space limitations, we do not present a formal cost model for this process. Our main focus is on minimizing the above mentioned two dominant costs: \textit{algorithm time} and \textit{index IO time}. 
We want to store the most important hash buckets from the cache to maximize total number of buffer hits.

\begin{table}[t]
	\centerline{
		\begin{tabular}{|c|c|c|c|}\hline
			{\bf } & {\bf Total} & {\bf AlgTime} & {\bf IndexIOTime}  \\ \hline\hline
			NS1: LRU & 263.6 & 117.6 & 146.0 \\ \hline 
			NS2: Per-Bucket & 279.5 & 260.8 & 18.7 \\ \hline 
		\end{tabular}
	}
	\caption{Performance Comparison of Naive Strategies NS1 and NS2 (in sec)}\label{tab:bufferNaiveComp}
	\vspace*{-0.2in}	
\end{table}

\begin{figure}[t]
	\centering
	\begin{subfigure}[b]{0.31\textwidth}
		\centering
		{\setlength{\fboxsep}{0pt}
			\setlength{\fboxrule}{0.2pt}
			\includegraphics[width=\linewidth]{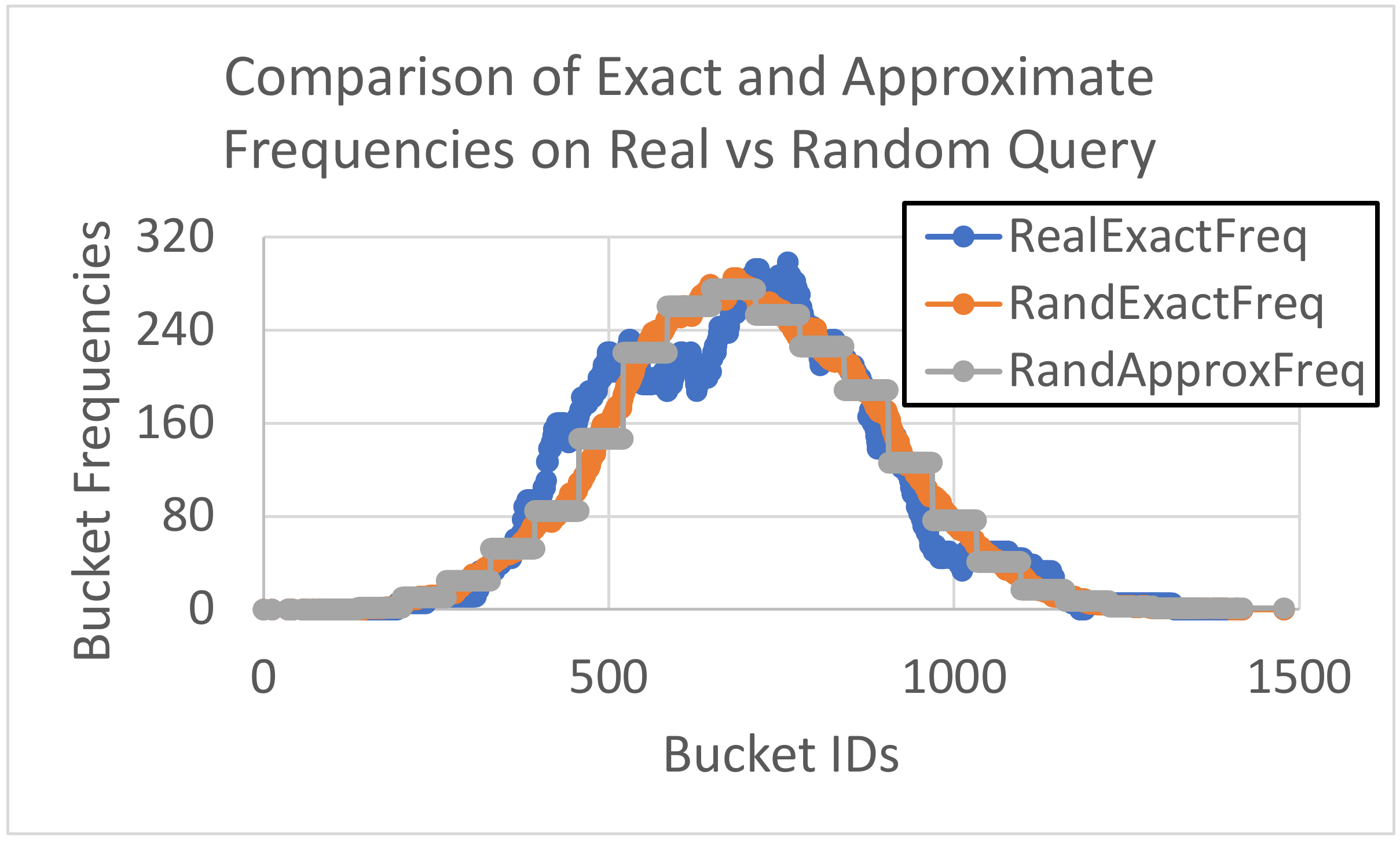}
		}
		
	\end{subfigure}\quad
	\begin{subfigure}[b]{0.31\textwidth}
		\centering
		{\setlength{\fboxsep}{0pt}
			\setlength{\fboxrule}{0.2pt}
			\includegraphics[width=\linewidth]{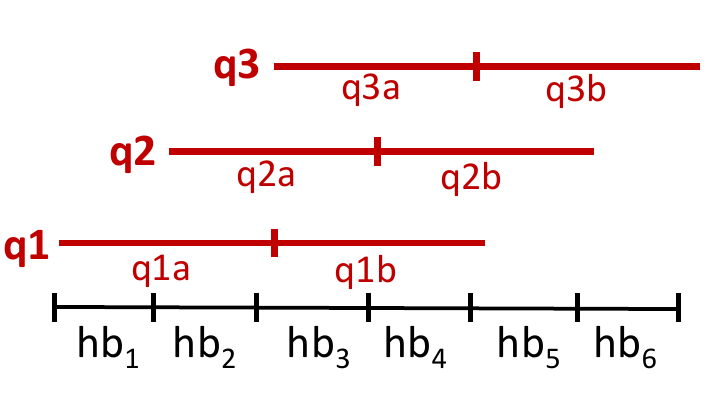}
		}
	\end{subfigure}\quad
	\begin{subfigure}[b]{0.31\textwidth}
		\centering
		{\setlength{\fboxsep}{0pt}
			\setlength{\fboxrule}{0.2pt}
			\includegraphics[width=\linewidth]{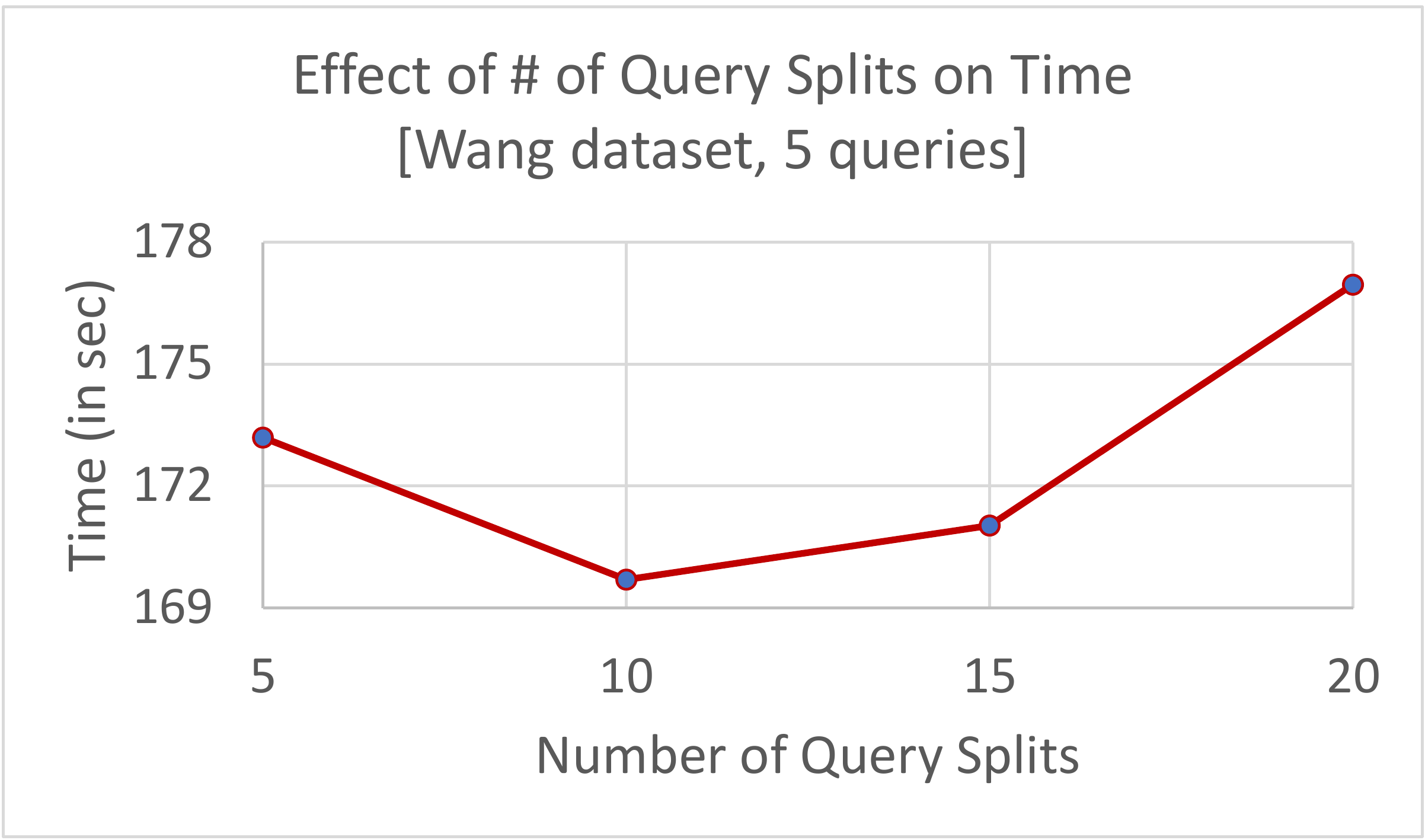}
		}
	\end{subfigure}
	%\vspace*{-0.15in}
	\caption{(a) Query Split Strategy, (b) Effect of \# of Query Splits on Time, (c) Comparison of Exact and Approx. Frequencies on a Real and Random Query}\label{fig:querySplitComp}
	\vspace*{-0.2in}
\end{figure}

\noindent\textbf{Naive Strategy 1: Using LRU Eviction Strategy on a given buffer.} Given $Q$, we first find the hash bucket locations for each of point queries of $Q$. In order to make the LRU (Least Recently Used) eviction strategy more effective, in each hash function $m$, we order the execution of point queries of $Q$ according to the hash bucket locations from left to right. During query processing, we evict the LRU index files from the buffer when the buffer gets full. 
	
\noindent\textbf{Naive Strategy 2: Using a Per-bucket Execution Strategy.} Since one of our goals is to reduce the $indexIOCost$, we also consider a Per-bucket execution strategy. Given a query object $Q$,
	we bring each useful hash bucket, $hb$, into the buffer, and for every $q$ in $Q$ that requires $hb$, we perform \textit{Collision Counting} to find the candidate nearest neighbor points to the point query. In Figure \ref{fig:querySplitComp} (b), this strategy would bring in $hb_1$ (then solve for q1), then bring $hb_2$ (and then solve for q1 and q2, since both queries using $hb_2$) and so on. 
	
As seen from Table \ref{tab:bufferNaiveComp}, NS1, due to its simplicity, has a lot smaller \textit{AlgTime} than NS2, but the \textit{IndexIOTime} of NS1 is a lot more than that of NS2. 
NS2 needs to find the queries that require the particular hash bucket brought into the main memory. While this process can be sped up with more index structures, it is still an expensive operation to check for all queries for each bucket, in each projection, for each radius. In each projection (in each radius), since a hash bucket is brought into the buffer only once for NS2, \textit{IndexIOTime} is the lowest. 

Hence we propose an efficient and effective buffer-conscious strategy that reduces the \textit{IndexIOTime} of NS1 without adding significant overhead to the \textit{AlgTime} (thus resulting in lower total time). Instead of using LRU, our eviction strategy is to evict a bucket based on the following three intuitive criteria. \textbf{Criterion 1}: if the bucket was not added to the buffer \textit{very recently}. When a bucket is added to the buffer, then there is a high likelihood that another query might use it in the near future.  
\textbf{Criterion 2}: if the bucket is \textit{far away} from the current query. It is more beneficial to evict another bucket that is far apart in the projection than the position of the current query.  \textbf{Criterion 3}: if the number of queries that still require this bucket (called \textit{frequency} of the bucket) is the lowest \textit{after} the first two criteria are satisfied. Criterion 3 ensures that a bucket needed by a lot of queries is not evicted. Due to space limitations, we do not formally show the pseudo-code for the eviction process.  

The main challenge in the above criteria is that the main criterion (Criterion 3) requires \textit{mmLSH} to know the frequencies of each bucket in each projection at each radius to decide which bucket to evict. This is an unfair expensive requirement to have during online query processing. Across different real multimedia datasets, we \textit{observed} that the frequencies of buckets on collection of queries associated with an object showed a behavior very similar to a collection of randomly chosen queries. Figure \ref{fig:querySplitComp} (a) shows that the bucket frequencies for a randomly chosen query from the Wang \cite{wangdataset} dataset exhibit a similar pattern for a set of randomly generated point queries on a single projection. 

\noindent\textbf{Projection-Dividing Strategy:} We use the above stated important observation to estimate the frequencies of buckets during \underline{offline processing}. The following is the overview: 1) \textit{We divide a projection into different regions}. Too few divisions will result in a high error between the estimated and actual frequencies. Too many divisions will also result in a high error because if the frequency behavior is slightly deviated than the random queries' behavior, then we assign same frequencies as that of the random queries. For this paper, we empirically decide the total number of divisions (set to 10).
2) \textit{We calculate the average frequencies for the random point queries for each region}, and assign the region's frequency to each bucket in that particular region. 
3) For each projection, \textit{we assign approximate frequencies to all buckets in each projection}.

\noindent\textbf{Query-Splitting Strategy:} In order to utilize the buffer more effectively, we split the queries into multiple sub-queries and reorder the execution of the queries based on these new set of queries. In Figure \ref{fig:querySplitComp} (b), the query execution order will change from $q1, q2, q3$ to $q1a, q2a, q1b, q3a, q2b, q3b$ to utilize the buffer more effectively. Note that too many splits is still detrimental due to the increase in the overall Algorithm time (like Naive Strategy 2).  Figure \ref{fig:querySplitComp} (c) shows the effect of different number of splits on the overall time. In this work, we empirically find a good split (that is found during the indexing phase, and set to 10). We leave finding the optimal split using advanced cost models to future work. 

\subsection{Theoretical Analysis}\label{sec:the}

\subsubsection{Guarantees on the Stopping Conditions}
\label{sec:guarantee}

The goal of this section is to prove the following theorem which provides a theoretical guarantee to \textit{mmLSH}. For simplicity we perform the theoretical analysis for the case $k=1$, the general case  follows similarly after simple  adaptations.

\begin{theorem}\label{mainTh}
	Let $Q$ be a query object and let $L=\min\{|X|\mid X\in \mathcal{D}\}.$ If $$\Gamma \geq \sqrt{ \max\left\{\frac{\ln \frac{1}{\delta}}{(\varepsilon-\delta)^2|Q|L},\, \frac{2\ln \frac{2}{\beta}}{\beta^2|Q|L} \right\}},$$
	then  \textit{mmLSH} finds a $\Gamma$-$c^2$-approximate NN with constant high probability.
\end{theorem}

For the proof of this theorem we  need the following lemma. For this, we   consider the following two properties for a given query object $Q$ and level $R$:

\begin{enumerate}[label=$\mathcal{P}\arabic*)$]
	\item  If $X$ is an object such that $\Gamma dist(Q,X)\leq R$ then $X$ is a $\Gamma$-candidate.
	
	\item The number of $\Gamma$-false positives is at most $\beta S$.
\end{enumerate}

\noindent In the next lemma, we show that the above properties hold with high probability.

\begin{lemma}\label{theLemma}
	Let  $\delta$ be the probability defined in Section \ref{sec:prelim} and $\varepsilon>\delta$ as defined in Section \ref{sec:keydef}, then if $\Gamma$ satisfies the inequality in Theorem \ref{mainTh} we have $Pr[\mathcal{P}1]\geq 1-\delta$ and $Pr[\mathcal{P}2]> \frac{1}{2}$.
\end{lemma}	
\begin{proof}
	For  $q_i\in set(Q)$ and $x_j\in set(X)$, let $A$  be the condition $cc(q,x_j)\geq l$, $B$ be  $||q,x_j||\leq R$, and $C$  be  $||q,x_j||> cR$. From the proof of  Lemma 1 in ~\cite{Gan:2012:LHS:2213836.2213898} we know the following inequalities hold:
	\begin{equation}\label{ins}
	Pr[A|B]\geq  1-\delta\quad \mbox{and} \quad Pr[\neg A | C]\geq (1-\exp(-2(\alpha-p_2)^2m))\geq (1-\frac{\beta}{2}).
	\end{equation}
	We proceed to prove  inequality $Pr[\mathcal{P}1]\geq 1-\delta$. Assume $\Gamma dist(Q,X)\leq R$, which is equivalent to 
 $Pr[||q_i,x_j||\leq R]\geq \Gamma,$
	where $q_i\in set(Q)$ and $x_j\in set(X)$.  Therefore, 
	$
	p=Pr[A]\geq Pr[A \wedge B]= Pr[A| B] Pr[ B]\geq (1-\delta)\Gamma,
	$
	where the last inequality follows from the left hand side inequality in Equation \eqref{ins}.
	
	For every $1\leq i\leq |Q|$ and $1\leq j\leq |X|$, let $Y_{i,j}\sim Ber(1-p)$ be a Bernoulli random variable which is equal to $1$ if $cc(q_i,x_j)< l$. Then 
	\begin{align*}
		Pr[ci(Q,X)\geq (1-\varepsilon)\Gamma]&=1-Pr[\sum_{i,j} Y_{i,j}\geq (1-(1-\varepsilon)\Gamma) |Q| |X|]\\
		&\geq 1-\exp(-2(\varepsilon-\delta)^2\Gamma^2)|Q||X|,
	\end{align*}
	where the  inequality follows from Hoeffding's Inequality. Therefore for the given range of $\Gamma$ we have
	$Pr[\mathcal{P}1]=Pr[ci(Q,X)\geq (1-\varepsilon)\Gamma]\geq 1-\delta.$
	
	We  continue with  the  proof of $Pr[\mathcal{P}2]> \frac{1}{2}$. For this, we assume $\Gamma dist(Q,X)>cR$. Which is equivalent to
	$Pr[||q_i,x_j||> cR]\geq 1-\Gamma.$
	Then 
		$$
		1-p=Pr[\neg A]\geq Pr[\neg A \wedge C]= Pr[\neg A| C] Pr[ C]\geq (1-\frac{\beta}{2})(1-\Gamma)
		$$ 
	where the last inequality follows from  the right hand side inequality in Equation \eqref{ins}. Therefore, 
	$p\leq \Gamma+\frac{\beta}{2}-\frac{\beta\Gamma}{2}.$ 
	For every $1\leq i\leq |Q|$ and $1\leq j\leq |X|$, let $Y_{i,j}\sim Ber(1-p)$ be a Bernoulli random variable defined as above. Thus 
	 $$
		Pr[ci(Q,X)\geq \Gamma+ \frac{\beta}{2}]
		=Pr[\sum_{i,j} Y_{i,j}\leq  (1-\Gamma-\frac{\beta}{2}-\Delta) |Q| |X|]
		$$ 
	for some $\Delta>0$. Thus, from Hoeffding's Inequality  it follows that 
$$
		q = Pr[ci(Q,X)\geq \Gamma+ \frac{\beta}{2}]
		< \exp(-2(\Gamma+\frac{\beta}{2}-p)^2|Q||X|)
		\leq \exp(-2\big(\frac{\beta}{2}\big)^2\Gamma^2|Q||X|).
		$$
	Let $FP$ be the set of false positives, that is 
	$FP=\{X\in \mathcal{D}\mid ci(Q,X)\geq \Gamma+ \frac{\beta}{2}\text{ and }\Gamma dist(Q,X)>cR\},$
	then  $Pr[\mathcal{P}2]=Pr[|FP| \leq \beta S]$. Therefore, it suffices to show the latter is larger than $\frac12$.
	
	Let $X_1,\ldots, X_S$ denote the elements of $\mathcal{D}$. For every $1\leq i\leq S$ let  $Z_i\sim Ber(q)$ be  the Bernoulli random variable  which is equal to one if $X_i\in FP$. Then the expected value of the size of $FP$ satisfies
	\begin{align*}
		E(|FP|)=E(\sum_i Z_i)&=\sum_i E(Z_i)
		=S\cdot q<S\cdot \exp(-2\big(\frac{\beta}{2}\big)^2\Gamma^2|Q||X|).
	\end{align*}
	Therefore, from Markov's Inequality it follows that 
	\begin{align*}
		Pr[|FP|]\leq \beta S]
		1-\geq\frac{E[|FP|]}{\beta S}
		> 1 -\frac{1}{\beta}\exp(-2\big(\frac{\beta}{2}\big)^2\Gamma^2|Q||X|)\geq \frac{1}{2},
	\end{align*}
	where the last inequality holds by the assumption on $\Gamma$. This finishes the proof.
\end{proof}
We are now ready to prove the theorem.
\begin{proof}[of Theorem \ref{mainTh}]
	By Lemma \ref{theLemma}  properties $\mathcal{P}1$ and $\mathcal{P}2$ hold with constant high probability. Therefore, we may  assume these properties hold  simultaneously.
	
	Let $r$ be the smallest $\Gamma$-distance between $Q$ and an object of $\mathcal{D}$. Set   $t = \lceil \log_c r\rceil$ and $R = c^t$.
	
	Assume first that the algorithm finishes with  terminating  condition $\mathcal{T}1$, that is at level $R$ at least $1+\beta S$ $\Gamma$-candidates have been found. By property $\mathcal{P}2$ at most $\beta S$ of these are false positives. Let $X$ be the object returned by the algorithm, then we have $\Gamma dist(Q,X)\leq cR\leq c^2r$.
	
	Now, if the algorithm does not finish with $\mathcal{T}1$, then property $\mathcal{P}1$ guarantees it finishes with $\mathcal{T}2$ at the end of level $R$.  Let $X$ be the object returned by the algorithm, then we have $\Gamma dist(Q,X)\leq R\leq cr<c^2r$. This finishes the proof.
\end{proof}

\section{Experimental Evaluation}
\label{sec:exp}
In this section, we evaluate the effectiveness of our proposed index structure, \textit{mmLSH} on four real multimedia data sets, under different system parameters. All experiments were run on the nodes of the Bigdat cluster \footnote{Supported by NSF Award \#1337884} with the following specifications: two Intel Xeon E5-2695, 256GB RAM, and CentOS 6.5 operating system. We used the state-of-the-art C2LSH \cite{Gan:2012:LHS:2213836.2213898} as our base implementation.\footnote{\textit{mmLSH} can be implemented over any state-of-the-art LSH technique.} All codes were written in C++11 and compiled with gcc v4.7.2 with the -O3 optimization flag. For existing state-of-the-art algorithms (C2LSH and QALSH), we used the Borda Count process (Section \ref{sec:mmLSH}) to aggregate the results of the point queries to find the nearest neighbor objects. Additionally, since the accuracy and the performance of the aggregation is affected by the chosen number of top-$k'$ results of the point queries, we choose a varying $k'$ for Linear, C2LSH, and QALSH for a fair comparison: $k' = 25, 50, 100$. We also implement an LRU buffer for the indexes in C2LSH and QALSH to show a fair comparison with our results. We compare our work with the following alternatives: 
\begin{itemize}[leftmargin=*]
	\item \textbf{LinearSearch-Borda:} In this alternative, the top-$k'$ results of the point queries are found using a brute-force linear search. This method does not utilize the buffer since it does not have any indexes.  
	\item \textbf{C2LSH-Borda:} top-$k'$ results of point queries are found using C2LSH \cite{Gan:2012:LHS:2213836.2213898}.
	\item \textbf{QALSH-Borda:} top-$k'$ results of point queries are found using QALSH \cite{Huang:2015:QLH:2850469.2850470}.	
\end{itemize}

%\vspace*{-0.15in}
\subsection{Datasets}
We use the following four real multimedia datasets to evaluate \textit{mmLSH}. Different feature extraction algorithms are used to show the effectiveness of \textit{mmLSH}.
\begin{itemize}[leftmargin=*]
	\item 
	\textbf{Caltech}\cite{Caltech256Image} This dataset consists of 3,767,761 32-dimensional points that were created using BRIEF on 28,049 images belonging to 256 categories.
	\item 
	\textbf{Corel}\cite{Corel10kImage} This dataset consists of 1,710,725 64-dimensional points that were created using SURF on 9,994 images belonging to 100 categories.
	\item
	\textbf{MirFlicker}\cite{MirFlicker25k} This dataset consists of 12,004,143 32-dimensional points that were created using ORB on 24,980 images.
	\item
	\textbf{Wang}\cite{wangdataset} This dataset consists of 695,672 128-dimensional SIFT descriptors belonging to 1000 images. These images belong to 10 different categories. 
\end{itemize}

\subsection{Evaluation Criteria and Parameters}
\label{sec:evalCriteria}

We evaluate the execution time and accuracy using the following criteria:
\begin{itemize}[leftmargin=*]
	\item \textbf{Time: }The two main dominant costs in LSH-based techniques are the algorithm time and the index IO time. We observed that the index IO times were not consistent (i.e. running the same query multiple times, which needed the same index IOs, would return drastically different results, mainly because of disk cache and instruction cache issues). Thus, the overall execution time is modeled for an HDD where an average disk seek requires 8.5 ms and an average data read rate is 0.156 MB/ms \cite{SeagateHDD}. 
	\item \textbf{Accuracy: }Similar to the ratio defined in earlier works \cite{Gan:2012:LHS:2213836.2213898,Huang:2015:QLH:2850469.2850470}, we define an object ratio to calculate the accuracy of the returned top-$k$ objects as following:
	$OR_{\Gamma}(Q) = \frac{1}{k} \sum_{i=1}^{k}\frac{\Gamma dist(Q, X_i)}{\Gamma dist(Q, X_i^*)}$
	where $X_1, ..., X_k$ denote the top-$k$ objects returned from the algorithm and $X_1^*, ..., X_k^*$ denote the real objects found from the ground truth. $\Gamma dist$ is computed using Equation \ref{eqn:gamma_dist}. Object Ratio of 1 means $100\%$ accuracy and as it increases, the accuracy decreases.
\end{itemize}

\noindent We do not report the index size or the index construction cost, since they would be the same as the underlying LSH implementation that we use (C2LSH \cite{Gan:2012:LHS:2213836.2213898}). 
\noindent 
We choose $\delta=0.1$, $\beta = \frac{25}{S}$, $\varepsilon=0.2$, $w=2.184$ \cite{Huang:2015:QLH:2850469.2850470} for C2LSH and mmLSH, $w=2.7191$ \cite{Huang:2015:QLH:2850469.2850470} for QALSH. 
We randomly chose 10 multimedia objects as queries from each dataset and report the average of the results. 

\begin{figure*}[h]
	\centering
	\begin{subfigure}[b]{0.4\textwidth}
		\centering
		{\includegraphics[width=\linewidth]{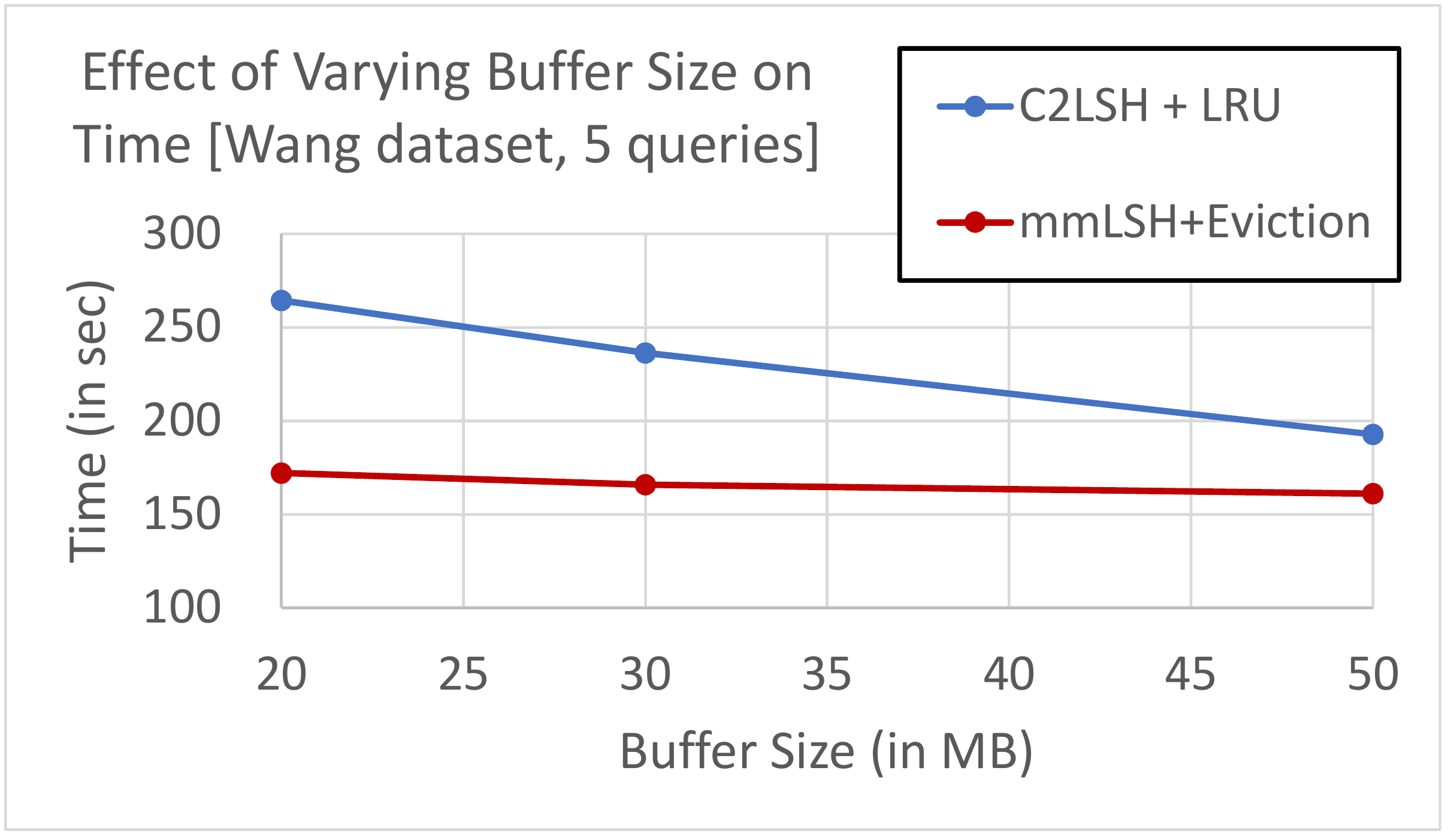}}
	\end{subfigure}\quad
	\begin{subfigure}[b]{0.4\textwidth}
		\centering
		{\includegraphics[width=\linewidth]{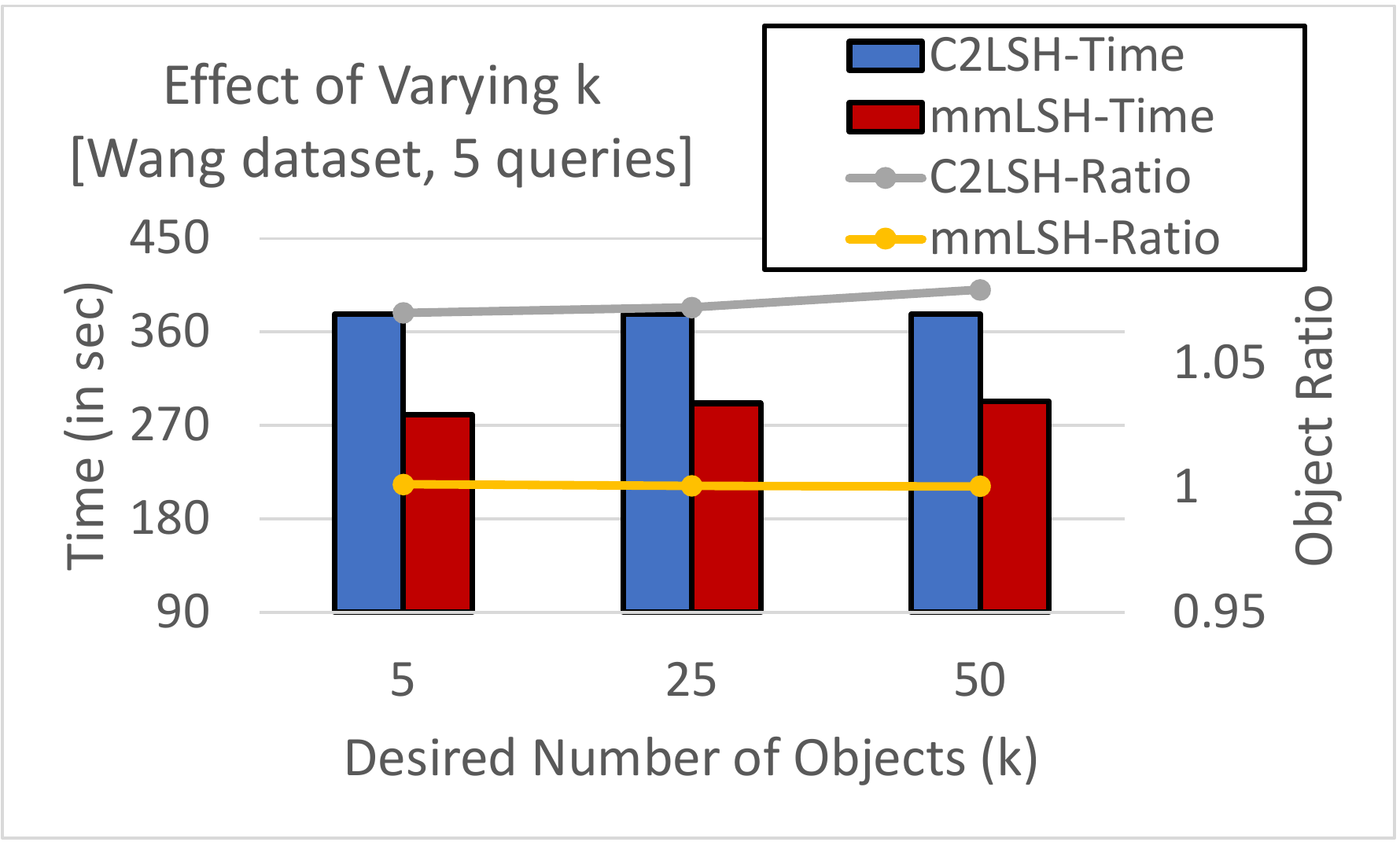}}
	\end{subfigure}\quad
	%\vspace*{-0.15in}
	\caption{Effect of (a) Buffer Size on Time, (b) Varying $k$ on Time and Accuracy}
	\label{fig:effectParams}
	%\vspace*{-0.4in}
\end{figure*}

\begin{figure*}[!h]
	\centering
	{\includegraphics[width=\linewidth]{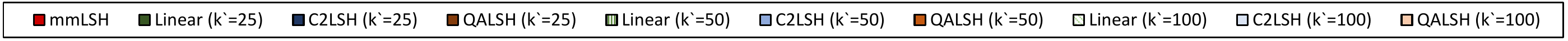}}
	
	\begin{subfigure}[b]{0.47\textwidth}
		\centering
		{\includegraphics[width=\linewidth, height=0.65in]{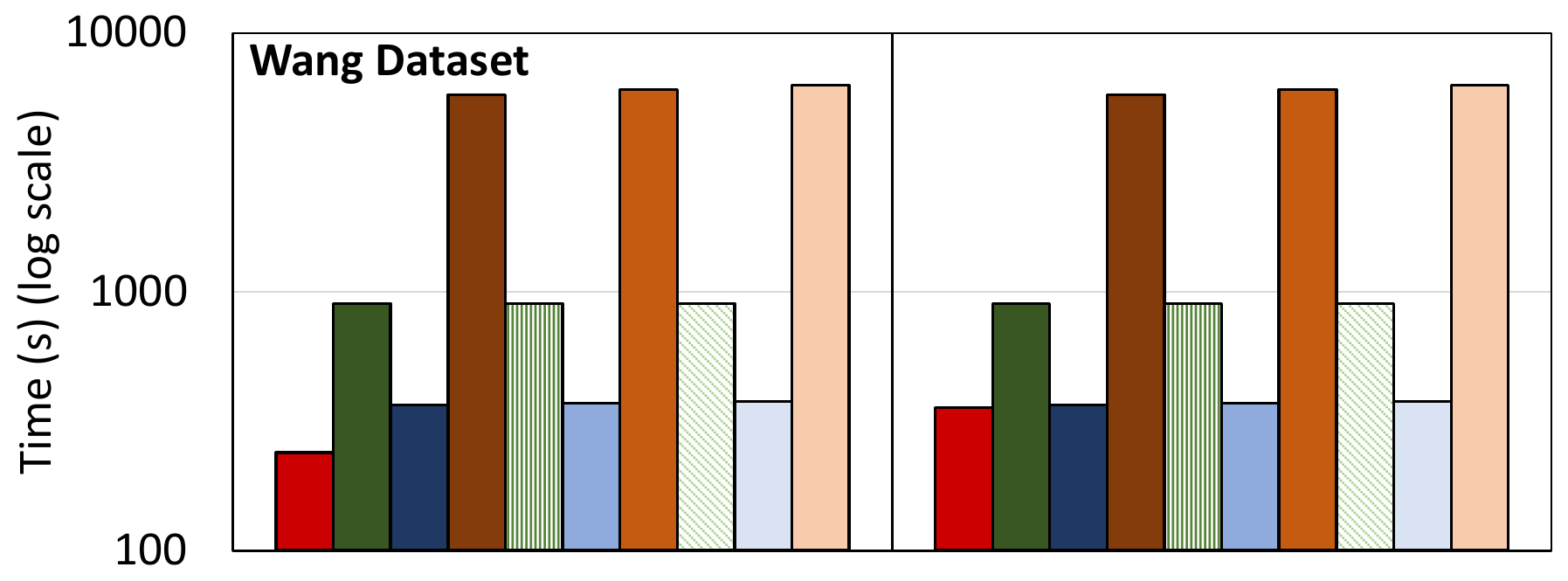}}
	\end{subfigure}\quad
	\begin{subfigure}[b]{0.47\textwidth}
		\centering
		{\includegraphics[width=\linewidth, height=0.65in]{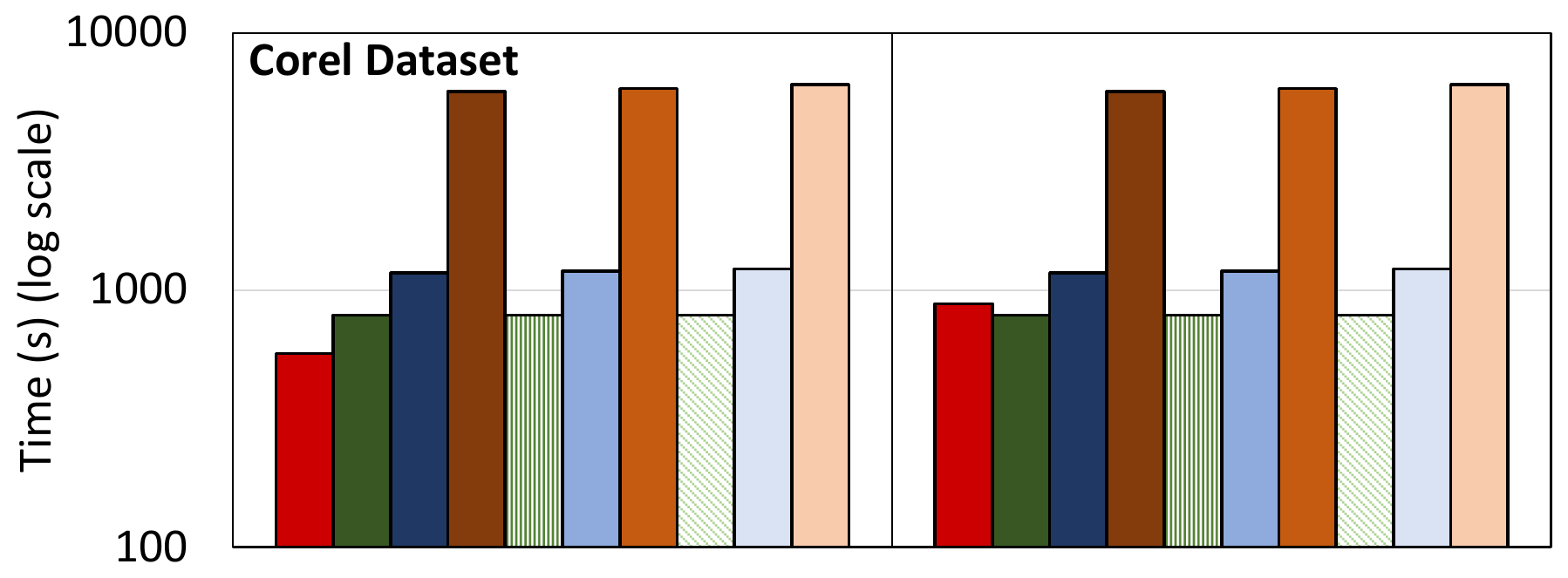}}
	\end{subfigure} \\

	\begin{subfigure}[b]{0.47\textwidth}
		\centering
		{\includegraphics[width=\linewidth, height=0.75in]{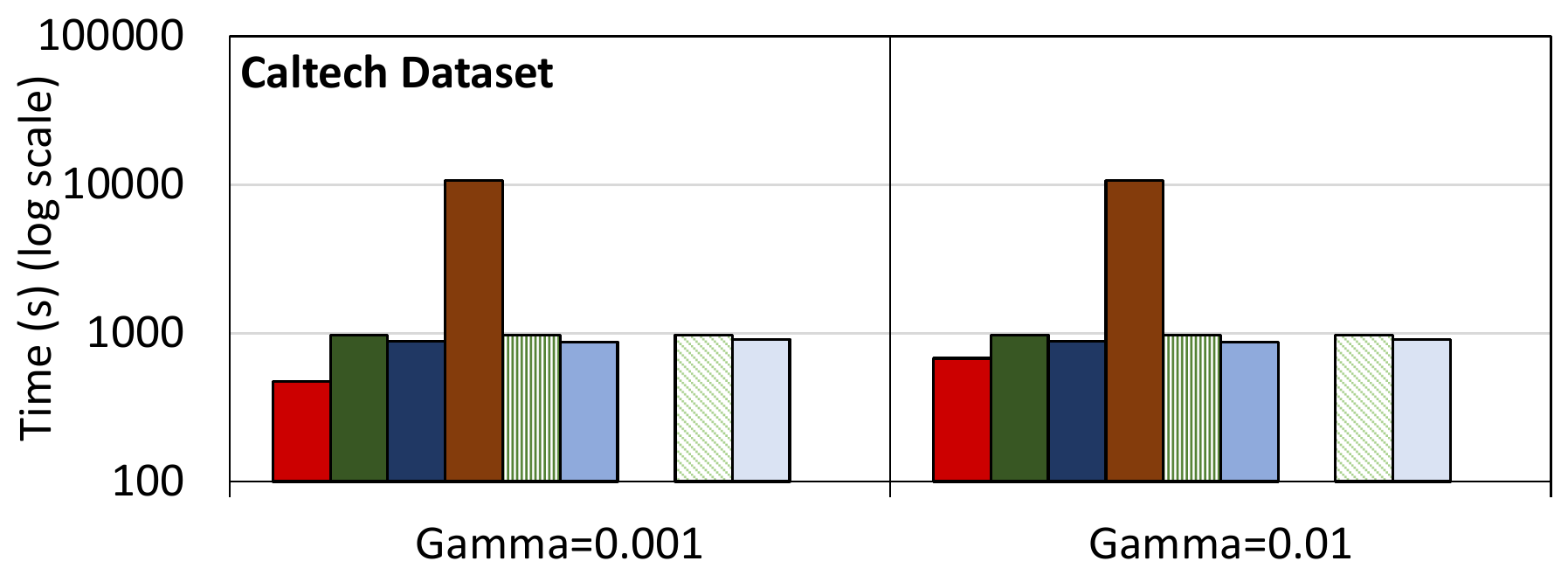}}
	\end{subfigure}\quad
	\begin{subfigure}[b]{0.47\textwidth}
		\centering
		{\includegraphics[width=\linewidth, height=0.75in]{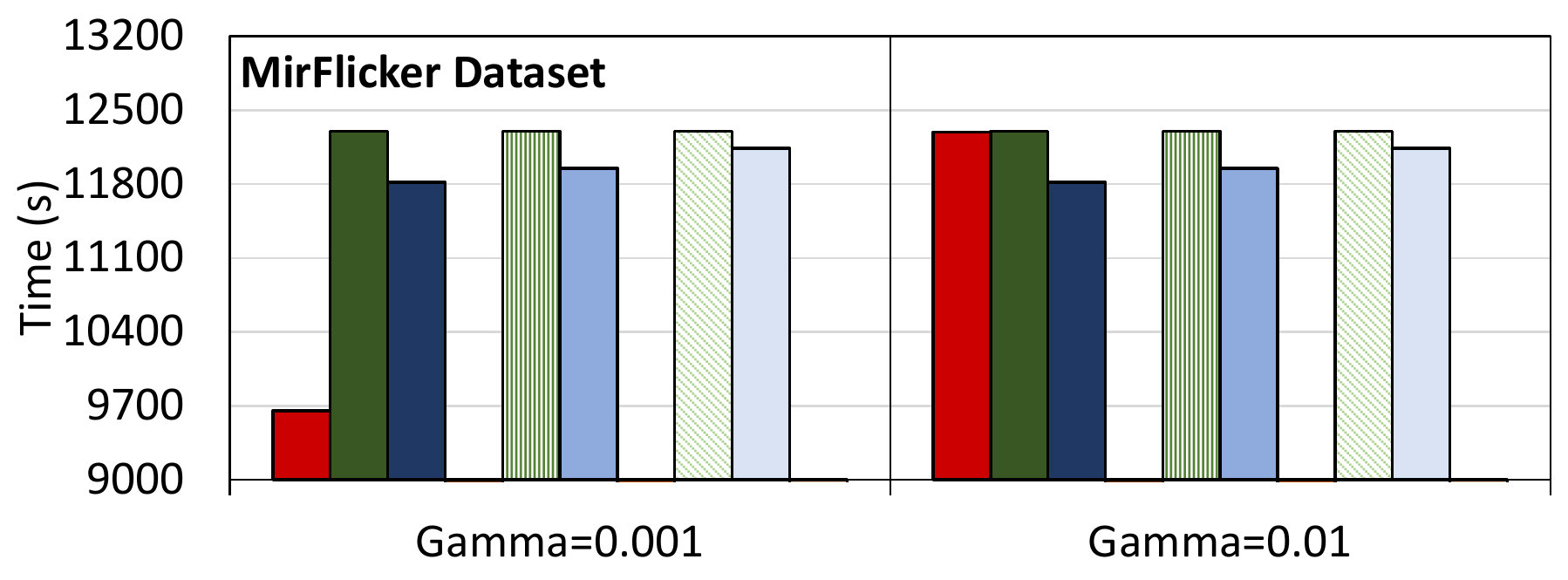}}
	\end{subfigure}

	\caption{Comparison of Time of \textit{mmLSH} against alternatives}
	\label{fig:expTime}

	\bigskip

	\begin{subfigure}[b]{0.47\textwidth}
		\centering
		{\includegraphics[width=\linewidth, height=0.65in]{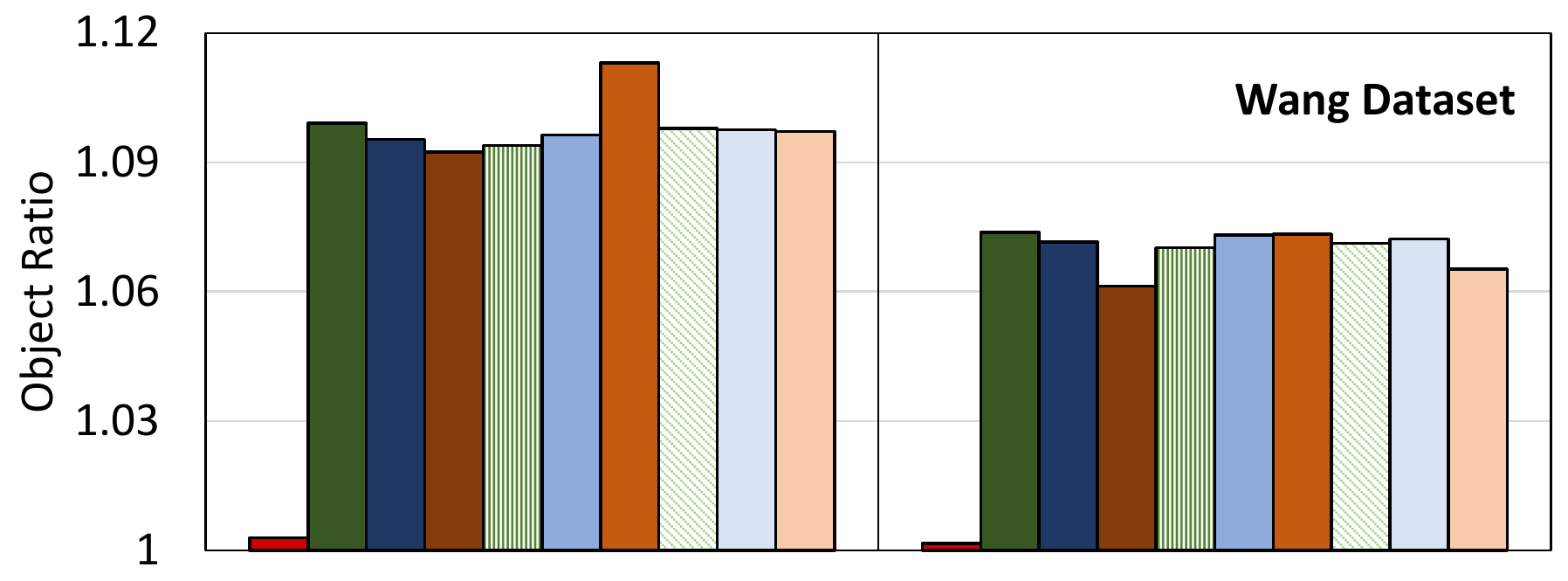}}
	\end{subfigure}\quad
	\begin{subfigure}[b]{0.47\textwidth}
		\centering
		{\includegraphics[width=\linewidth, height=0.65in]{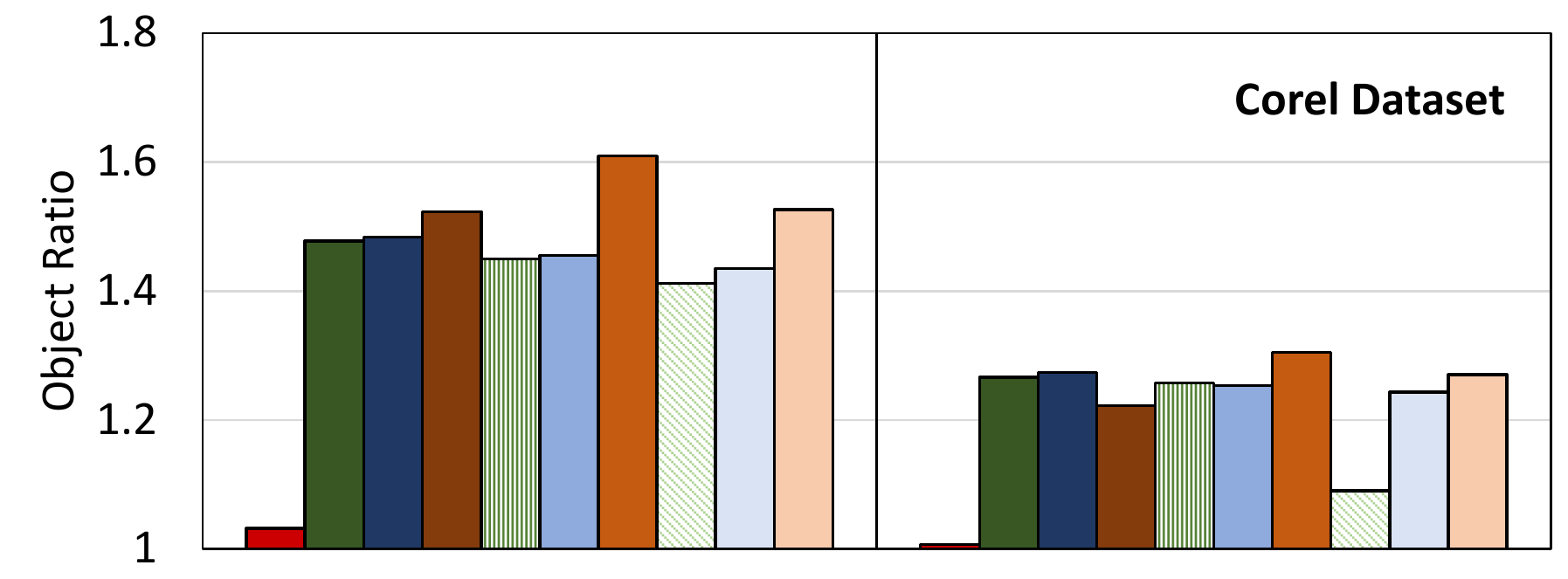}}
	\end{subfigure} \\
	
	\begin{subfigure}[b]{0.47\textwidth}
		\centering
		{\includegraphics[width=\linewidth, height=0.75in]{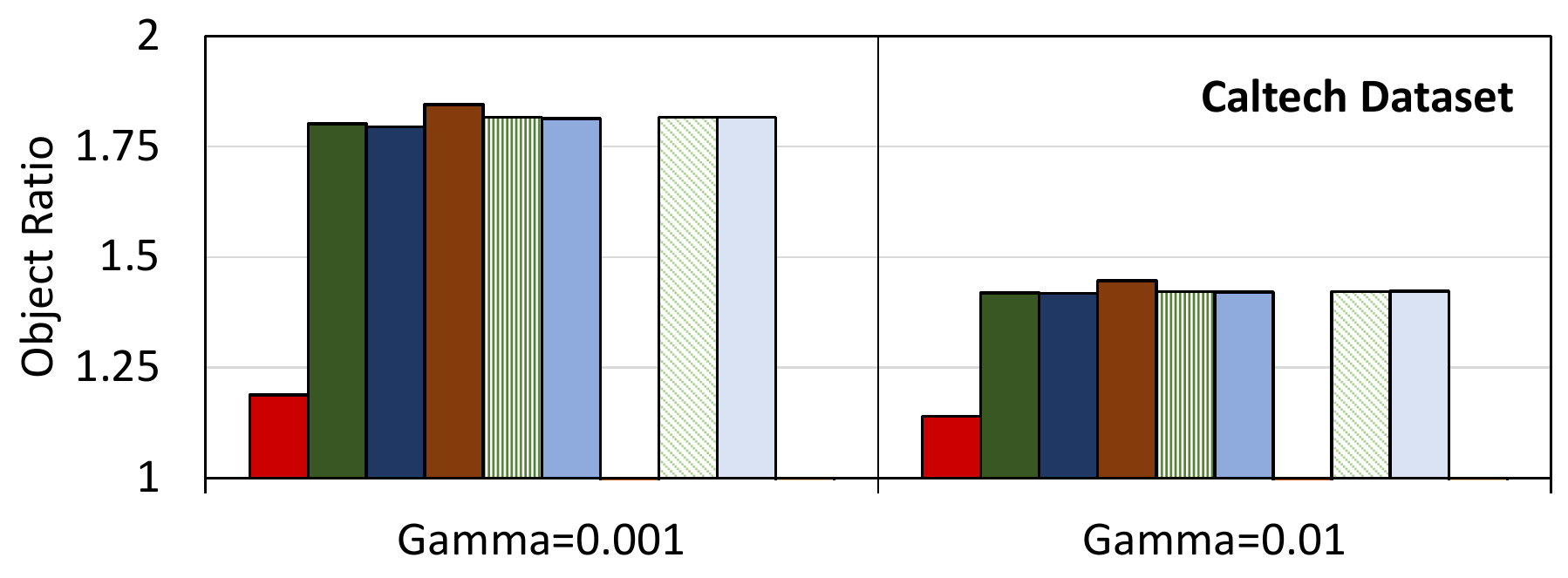}}
	\end{subfigure}\quad
	\begin{subfigure}[b]{0.47\textwidth}
		\centering
		{\includegraphics[width=\linewidth, height=0.75in]{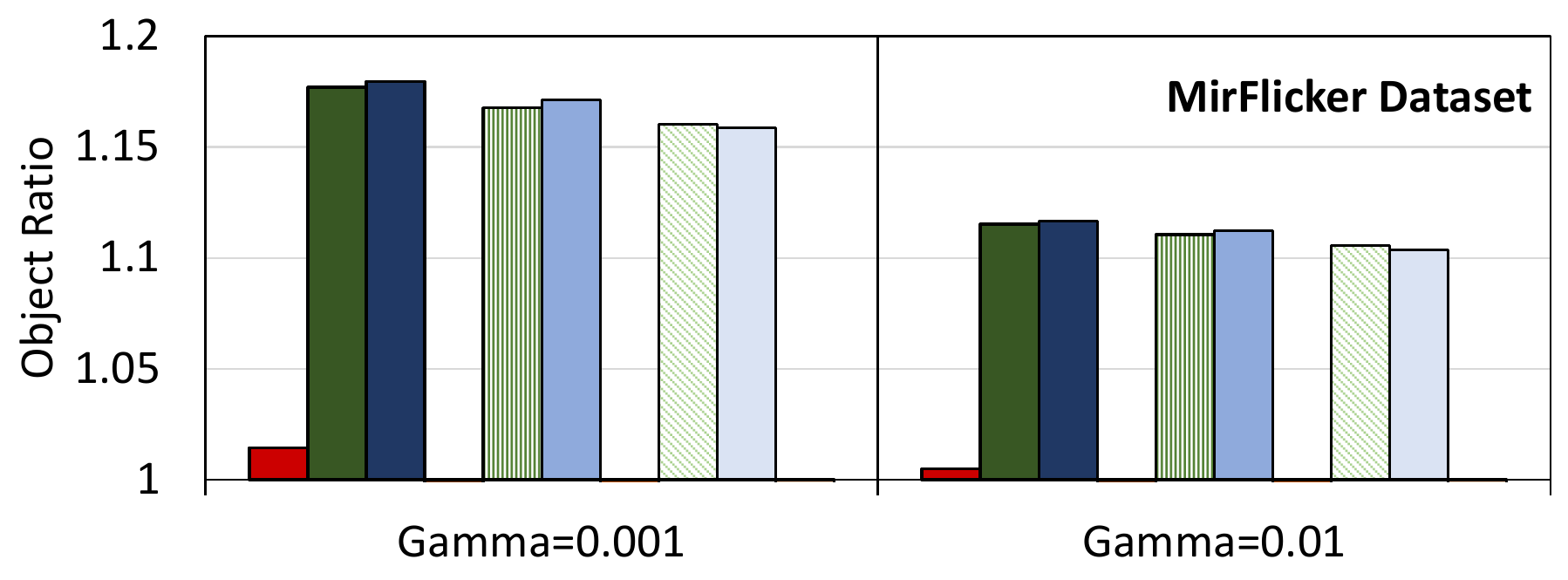}}
	\end{subfigure}
	
	\caption{Comparison of Accuracy of \textit{mmLSH} against alternatives}
	\label{fig:expRatio}
\end{figure*}

\subsection{Discussion of the Results}
\label{sec:results}

In this section, we analyze the execution time and accuracy of \textit{mmLSH} using the criteria explained in Section \ref{sec:evalCriteria} against its alternatives. We note that QALSH gives drastically worse times than C2LSH, and hence when comparing the effectiveness for varying parameters, we only compare with C2LSH. \\

\noindent\textbf{Effect of Buffer Size: }
Figure \ref{fig:effectParams} (a) shows the benefit of our eviction strategy (Section \ref{sec:bufferStrategy}) when compared with C2LSH + LRU for varying buffer sizes. It is evident from this figure that our three criterion are helpful in evicting less useful index files (and keeping the more useful files in the buffer). The small overhead in Algorithm time is offset by a significant reduction in the number of index IOs, which eventually results in lower overall time. We found out that the highest L3 cache size is 24.75 MB for desktop processors and 60MB for server processors. Therefore, we decided to choose the buffer sizes between 20 MB and 50 MB. We use 30 MB as the default cache size in the following experiments. \\

\noindent\textbf{Effect of Number of Desired Objects: }
Figure \ref{fig:effectParams} (b) shows the execution time and accuracy of \textit{mmLSH} against C2LSH for varying number of desired objects ($k$). This figure shows that \textit{mmLSH} has better time and object ratio for different $k$ values. Additionally, it shows that \textit{mmLSH} is scalable for a large number of desired objects as well. Moreover, although the object ratio of \textit{mmLSH} stays the same by increasing k, the object ratio of C2LSH increases. We use $k=25$ as the default for the following experiments. \\

\noindent\textbf{Comparison of mmLSH vs. State-of-the-art Methods}
Figures \ref{fig:expTime} and \ref{fig:expRatio} show the time and accuracy of \textit{mmLSH}, LinearSearch-Borda, C2LSH-Borda, QALSH-Borda for 4 multimedia datasets with varying characteristics. The Borda count process is done after query processing and takes very negligible time. Note that, In our work, we consider all feature-vectors that are extracted by a feature-extraction algorithm. Several works \cite{5523894,10.1007/978-3-642-13772-3_40} have been proposed that cluster these points with the purpose of finding a representative point to reduce the complexity and overall processing time of the problem. Our work is orthogonal to those approaches and hence are not included in this paper. 

We found that QALSH took a very long time as compared with \textit{mmLSH} and other alternatives. For the Caltech and MirFlickr datasets, QALSH did not finish the experiments due to their slow execution and hence are not included in the charts. The slow execution is mainly due to the use of the B+-tree index structures to find the nearest neighbors in the hash functions. \textit{mmLSH} always returns a higher accuracy than the alternatives while being much faster than all three alternatives. This is because \textit{mmLSH} is able to leverage the common elements between queries and improve cache utilization along with being able to stop earlier than the state-of-the-art algorithms. For future work, we plan on investigating the application of \textit{mmLSH} to other distance measures (such as Hausdorff distance, etc.) and compare with (and possibly utilize) other feature vector aggregation techniques \cite{Jegou:2010:IBL:1718320.1718326}.

\section{Conclusion}
\label{sec:concl}

In this paper, we presented a novel index structure for efficiently finding top-k approximate nearest neighbors for \underline{m}ulti\underline{m}edia data using \underline{LSH}, called \textit{mmLSH}. Existing LSH-based techniques can give theoretical guarantees on these individual high-dimensional feature vector queries, but not on the multimedia object query. These techniques also treat each individual feature vector belonging to the object as independent of one another. 
In \textit{mmLSH}, novel strategies are used that improve execution time and accuracy of a multimedia object query. Additionally, we provide rigorous theoretical analysis and guarantees on our returned results. Experimental evaluation shows the benefit of \textit{mmLSH} in terms of execution time and accuracy compared to state-of-the-art algorithms. Additionally, \textit{mmLSH} can give theoretical guarantees on the final results instead of the individual point queries.

% \bibliography{References}
\end{document}